\documentclass[]{article}
\usepackage{graphicx,subfigure}
\usepackage{amsmath,amsfonts,amssymb,amsbsy,amscd,amsgen}
\usepackage{amsthm}  % comment when using JFM
\usepackage[authoryear,round]{natbib}
\usepackage{color}
\usepackage{appendix}
\usepackage[colorlinks,citecolor=black]{hyperref}  % JFM and standard 
\usepackage{comment}
\usepackage{csquotes}
\usepackage{enumerate}
\usepackage{dsfont}

\title{Optimal initial condition of passive tracers for 
their maximal mixing in finite time}
\author{Mohammad Farazmand\thanks{Email address: mfaraz@mit.edu}\\
{\small Department of Mechanical Engineering,}
{\small Massachusetts Institute of Technology}\\
{\small77 Massachusetts Ave., Cambridge MA, USA}
}
\date{}
\newcommand{\id}{\mathrm d}
\newcommand{\vc}{\mathbf}
\newcommand{\bnabla}{\pmb\nabla}

\newcommand{\fm}{\varphi}
\newcommand{\spn}{\mbox{span}}
\newcommand{\proj}{\Pi}
\newcommand{\fopt}{f_{\mbox{\tiny opt}}}
\newcommand{\Vi}{V_{I}}
\newcommand{\Vnu}{V_{\nu}}

\renewcommand{\i}{\mbox{i}}

\theoremstyle{definition}
\newtheorem{thm}{Theorem}
\newtheorem{rem}{Remark}

\newtheorem{defn}{Definition}

\begin{document}
%%% PRS A
%\title{Optimal initial condition of passive tracers for 
%their maximal mixing in finite time}
%\author{M. Farazmand}
%\address{Department of Mechanical Engineering\\
%Massachusetts Institute of Technology\\
%77 Massachusetts Ave., Cambridge MA, USA}
%\subject{Chaos theory, Fluid mechanics}
%\keywords{chaotic advection, mixing, PF operator, functional analysis, optimization}
%\corres{mfaraz@mit.edu}
%\begin{fmtext}
%\end{fmtext}

\maketitle

\begin{abstract}
The efficiency of a fluid mixing device is often limited by fundamental 
laws and/or design constraints, such that a perfectly homogeneous 
mixture cannot be obtained in finite time.
Here, we address the natural corollary question: Given the best available mixer, what is
the optimal initial tracer pattern that leads to the most homogeneous mixture
after a prescribed finite time?
For ideal passive tracers, we show that this optimal initial condition coincides with the right singular vector 
(corresponding to the smallest singular value) of 
a suitably truncated Perron-Frobenius (PF) operator. 
The truncation of the PF operator is made under the assumption that there is
a small length-scale threshold $\ell_\nu$ under which the tracer blobs are considered, 
for all practical purposes, completely mixed. We demonstrate our results on two examples: a prototypical model
known as the sine flow and a direct numerical simulation of two-dimensional turbulence. 
Evaluating the optimal initial condition through this framework only requires the 
position of a dense grid of fluid particles at the final instance and their preimages at the initial 
instance of the prescribed time interval. As such, our framework can be 
readily applied to flows where such data is available through numerical simulations or
experimental measurements. 
\end{abstract}

\section{Introduction}
Given a fluid velocity field $\vc u(\vc x,t)$, a passive tracer satisfies the linear advection equation
\begin{equation}
\partial_{t}\rho +\vc u\cdot \bnabla \rho=0,\quad \rho(\vc x,t_0)=f(\vc x)
\label{eq:adveq}
\end{equation}
where the scalar field $\rho(\vc x,t)$ denotes the concentration of the tracer at time $t$
and $f$ is its initial concentration at time $t_0$. \citet{aref} pointed out that 
laminar unsteady velocity fields can, over time, develop complex tracer patterns consisting of ever smaller scales. 
This observation has inspired the successful development of many stirring protocols to enhance mixing 
in engineered devices (see, e.g.,~\cite{stroock02,gouillart06,mathew07,singh08,thiffeault08,gubanov10,foures14}).

Systematic classification of mixing efficiency of fluid flow, however, is relatively recent.
This classification was initiated by~\cite{lin11} who derived rigorous bounds on the mixing
efficiency of velocity fields with a prescribed stirring energy or stirring power budget. 
A notable outcome of their program is the rather remarkable discovery of a 
finite-energy velocity field ($\|\vc u\|_{L^2}=const.<\infty$) that achieves 
perfect mixing in finite time~\citep{lunasin12}. It was shown later, however, that any such velocity 
field must have infinite viscous dissipation, i.e. $\|\bnabla\vc u\|_{L^2}=\infty$~\citep{seis13,Iyer14}.

Besides this fundamental limitation, the implementation of mathematically obtained 
optimal stirring strategies in a mixing device is not always feasible due to, for instance, 
geometric constraints. 
The problem is more acute in natural fluid flow (such as geophysical flows or the blood stream)
over which we have virtually no control. 
%It is well-known that such flows, even at high Reynolds numbers, 
%contain long-lived coherent vortices~\citep{mcwilliams1984vortices,provenzale99}. 
%These vortices exhibit minimal material deformation
%over their lifetime, and therefore impede efficient mixing. 
%This is accentuated by recent revelations that some of these vortices
%remain perfectly coherent over long periods of time~\citep{bhEddy}.

In light of the above discussion, the natural question is: 
\begin{displayquote}
(Q) \emph{
Given an unsteady velocity field, what is the optimal initial tracer pattern
that leads to the most homogeneous mixture after a prescribed finite time?
}
\end{displayquote}
In spite of its importance, this question has received relatively little attention. \cite{hobbs97}
carried out a case study where the effect of the tracer injection location in a Kenics mixer is
examined. They find that, at least for short time horizons, the mixing efficiency depends
significantly on the injection location. A similar case study
is carried out by \cite{gubanov09} who studied the mixing efficiency of 
five different initial tracer patterns in a two-dimensional nonlinear model, known as the sine flow.

\cite{thiffeault08b} addressed an analogous question:
the asymptotic mixing of passive tracers advected under a steady velocity field
where the tracer is injected continuously into the flow via source terms. Through a variational approach, they
determined the optimal distribution of the sources (also see~\cite{okabe08}, for related numerical results).

Here, we address the finite-time mixing of passive tracers advected by fully unsteady
velocity fields, as formulated in~\eqref{eq:adveq}. 
Specifically, we seek the optimal initial condition $f$ that leads to the 
most homogeneous mixture after a given finite time. 
To the best of our knowledge, a rigorous method for determining this optimal initial condition 
is missing. 

Problem (Q) can, in principle, be formulated and solved as an infinite-dimensional
optimization problem, where the optimal initial condition coincides with the minimizer
of an appropriate cost functional. Such minimizers are typically 
obtained by iterative methods of adjoint-based 
optimization~\citep{protas08,faraz_cont}. This is, however, computationally prohibitive since it requires the backward-time 
integration of an adjoint partial differential equation (PDE) at each iteration.

Here, we show that under reasonable assumptions, the problem
reduces to a finite-dimensional one that can be readily solved at a relatively low computational cost.
To obtain this finite-dimensional reduction, we assume that tracer blobs
smaller than a small, prescribed length-scale $\ell_\nu$ are considered completely mixed for all practical 
purposes. This assumption, that is made precise in Section~\ref{sec:theory}, 
results in a natural Galerkin truncation of the Perron--Frobenius (PF) operator
associated with the advection equation~\eqref{eq:adveq}. We show that the optimal initial condition $f$
then coincides with a singular vector of the truncated PF operator. 

Our results complement the transfer operator-based methods for detecting finite-time 
coherent sets in unsteady fluid flows (\cite{froyland}; also see~\cite{dellnitz99,froyland07,froyland10,williams15}). 
Coherent sets refer to subsets of the fluid which exhibit minimal deformation under advection and therefore inhibit efficient mixing of tracers 
with the surrounding fluid. Our aim here is the opposite, namely, initially large-scale structures that 
under advection deform mostly into small-scale filaments.

The remainder of the paper is organized as follows. In section~\ref{sec:prelim}, we introduce some basic
notation and definitions. Section~\ref{sec:theory} contains our main results and
section~\ref{sec:numerics} details their numerical implementation. In section~\ref{sec:results}, the
results are demonstrated on two examples.

\section{Preliminaries}\label{sec:prelim}
Consider an unsteady, incompressible velocity field $\vc u(\vc x,t)$ defined over a 
bounded open subset $\mathcal D\subset \mathbb R^d$ where 
$d=2$ or $d=3$ for two- and three-dimensional flows, respectively.
The trajectories $\vc x(t;t_0,\vc x_0)$ of the fluid particles satisfy the ordinary differential equation
\begin{equation}
\dot{\vc x}=\vc u(\vc x,t),\quad t\in\mathbb R,
\label{eq:ode}
\end{equation}
where $\vc x(t;t_0,\vc x_0)$ denotes the time--$t$ position of the particle starting from the initial position $\vc x_0$ at time $t_0$.
If the velocity field is sufficiently smooth, there exists a two-parameter family of homeomorphisms
$\fm_{s}^t$ (the flow map)  such that $\vc x(t;s,\vc x_0)=\fm_{s}^t(\vc x_0)$ for all times $t$ and $s$.
As our interest here is in finite-time mixing, we restrict our attention to a prescribed finite time interval 
$[t_0,t_0+T]$ of interest. The flow map $\fm_{t_0}^{t_0+T}$ takes the initial position $\vc x_0$ 
of a fluid particle at time $t_0$ to its final position at time $t_0+T$.
Since the finite time interval is fixed, we drop the dependence of the flow map on $t_0$ and $t_0+T$, 
and write $\fm$ for notational simplicity.

Let $\rho(\vc x,t)$ denote the concentration of a passive tracer, 
i.e. $\rho$ satisfies equation~\eqref{eq:adveq}. 
Since the passive tracer is conserved along
fluid trajectories, we have
\begin{equation}
\rho(\vc x,t_0+T)=\rho(\fm^{-1}(\vc x),t_0)=f\circ \fm^{-1}(\vc x),
\label{eq:rhoIsPassive}
\end{equation}
for all $\vc x\in \mathcal D$. Note that since the flow map is a homeomorphisms, the inverse $\fm^{-1}$ is
well-defined. Equation~\eqref{eq:rhoIsPassive} motivates the definition of the Perron-Frobenius (PF) operator.

\begin{defn}[Perron--Frobenius operator]
The Perron--Frobenius operator associated with the flow map $\fm:\mathcal D\to \mathcal D$ is the linear transformation 
$\mathcal P:L^2(\mathcal D)\to L^2(\mathcal D)$ 
such that, for all $f\in L^2(\mathcal D)$,
\begin{equation}
(\mathcal P f)(\vc x)=f\circ \fm^{-1}(\vc x),\quad \forall\vc x\in\mathcal D.
\label{eq:PF}
\end{equation}
\label{def:PF}
\end{defn}

%Initially introduced by~\cite{koopman31}, the application of this operator has become 
%widespread in dynamical systems theory, and more recently in fluid dynamics 
%(see~\cite{budisic12} and~\cite{mezic13}, for reviews).
%We have slightly modified PF's original definition which read $\mathcal P f=f\circ\fm$.

The evolution of passive tracers can be described by the action of the PF operator on their initial 
conditions. More specifically, for the passive tracer $\rho$ described above, we have
\begin{equation}
\rho(\vc x,t_0+T)=(\mathcal Pf)(\vc x),
\end{equation}
for all $\vc x\in\mathcal D$ (cf. equation~\eqref{eq:rhoIsPassive}). 

We point out that there is a more general definition of the PF operator applicable to 
non-invertible dynamics (see Definition 3.2.3 of~\cite{mackey1994chaos}). 
In the special case where the flow map $\fm$ is invertible and volume-preserving, the general definition 
is equivalent to Definition~\ref{def:PF} above (Corollary 3.2.1 in~\cite{mackey1994chaos};~\cite{froyland2009}).

For incompressible flow, the 
PF operator is a unitary transformation with respect to the $L^2(\mathcal D)$ inner product 
$\langle\cdot,\cdot\rangle_{L^2}$. As a consequence, the $L^2$-norm $\|\rho(\cdot,t)\|_{L^2}$ of the tracer remains 
invariant under advection. Furthermore, the spatial average of the tracer is an invariant. 
Without loss of generality, one can assume that this spatial average vanishes,
$\int_{\mathcal D}\rho(\vc x,t)\id \vc x=0$ \citep{lin11}.

There has been several attempts to detect coherent structures in unsteady fluid flows using approximations of
the PF operator~\citep{froyland07,santitissadeekorn10,froyland10}. \cite{froyland} puts these approaches on a 
mathematically rigorous basis by composing the PF operator with diffusion operators. The resulting diffusive PF operator
is compact and has a well-defined singular value decomposition (SVD). \cite{froyland} shows that 
a singular vector, corresponding to the largest non-unit singular value of the diffusive PF operator, 
can reveal minimally dispersive subsets of the fluid that remain coherent and thereby inhibit
mixing (also see~\cite{froyland14}). 
Our goal here, however, is the opposite as we seek passive tracer initial conditions that mix 
most efficiently with their surrounding fluid. 

\section{Optimal initial conditions}\label{sec:theory}
\subsection{Physical considerations}
Given an initial tracer distribution, a reasonable mixer will generically
deform the tracer through stretching and folding of material elements
such that, over time, it develops ever smaller length scales.
It is, therefore, desirable to release the tracer initially into smallest possible scales.
In practice, the initially available range of scales into which the tracer may be released
is limited to relatively large scales. We denote this
large length-scale limit by $\ell_I$ (see figure~\ref{fig:schem_scales}, for an illustration).
It is left then to the fluid flow to transform the initially large-scale 
blobs of tracer to small filaments through a stretch-and-fold mechanism.
\begin{figure}
\centering
\includegraphics[width=.95\textwidth]{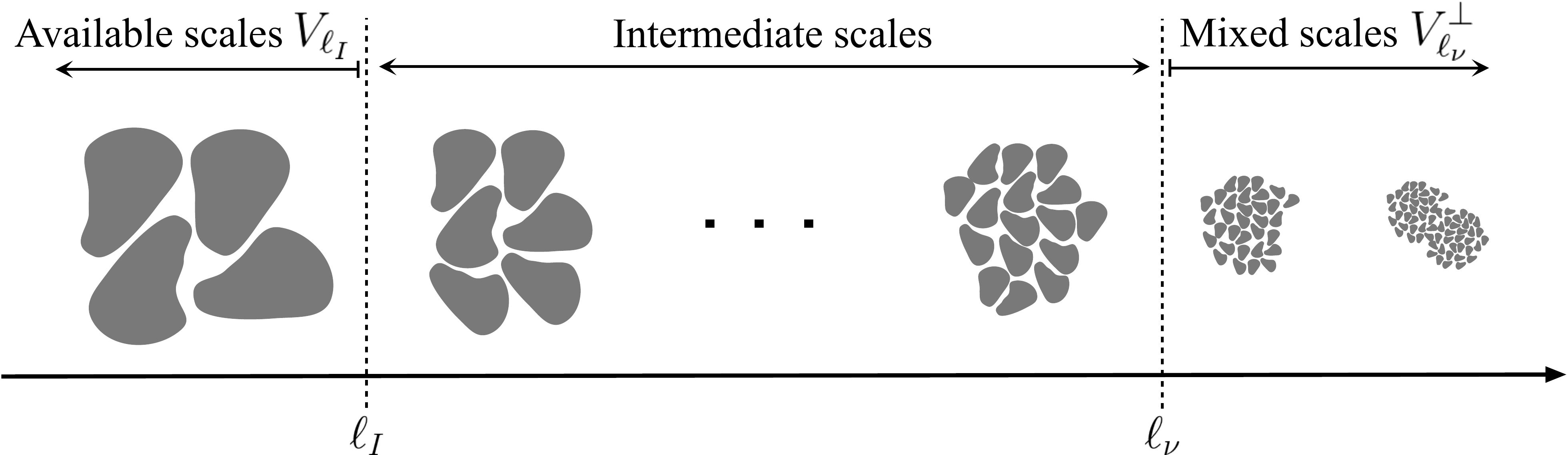}
\caption{An illustration of the initially available scales (larger than $\ell_I$)
and the mixed scales (smaller than $\ell_\nu$).}
\label{fig:schem_scales}
\end{figure}

On the other hand, we assume that there is a small length scale threshold $\ell_\nu\ll \ell_I$,
under which the tracer is considered, for all practical purposes, completely mixed. An efficient mixer, 
therefore, transfers the tracer distribution from large, initially available scales $\ell\geq\ell_I$
to the mixed scales $\ell<\ell_\nu$. In the next section, we make these statements precise.

\subsection{Mathematical formulation}
We consider a set of functions $\{\phi_j\}_{j\geq 1}$ forming a complete, orthonormal
basis for the space of square integrable functions $L^2(\mathcal D)$.
%\begin{comment}
That is $\langle\phi_i,\phi_j\rangle_{L^2}=\delta_{ij}$
and for any $f\in L^{2}(\mathcal D)$ there are constants $\alpha_j\in\mathbb R$ such that
$\lim_{k\to \infty}\|f-\sum_{j=1}^{k}\alpha_j\phi_j\|_{L^2}=0$.
%\end{comment}

We also assume that there is a length scale $\ell_j$ associated to each function $\phi_j$, and that
they are ordered such that the sequence $(\ell_1,\ell_2,\cdots)$ is decreasing. In other words, the length scale
associated with the function $\phi_j$ decreases as $j$ increases. Such a basis can be taken, for instance, 
to be Fourier modes or wavelets~\citep{walnut13}.

With this basis, we can mathematically model the subspace of initial conditions $\Vi$.
The subspace $\Vi$ consists of all scalar functions $f$ whose smallest length scale is larger than or equal to $\ell_I$. 
Since the basis $\{\phi_j\}_{j\geq 1}$ is ordered, there is a positive integers $n$ such that 
\begin{equation}
\Vi= \spn\{\phi_1,\phi_2,\cdots,\phi_n\}= \{\mbox{Initially available length scales}\ \ell\geq\ell_I \}.
\end{equation}

The subspace of unmixed length scales $\Vnu$ can be modeled similarly using 
a basis $\{\psi_i\}_{i\geq 1}$ for $L^2(\mathcal D)$. 
We assume that this basis is also orthonormal, complete and associated with a
decreasing sequence of length scales.
The subspace $\Vnu$ consists of all scalar functions whose smallest length scale is larger than or equal to
the unmixed length scale $\ell_\nu$. Therefore, there is $N\gg n$ such that 
\begin{equation}
\Vnu= \spn\{\psi_1,\psi_2,\cdots,\psi_N\}=\{\mbox{Unmixed length scales}\ \ell\geq\ell_\nu \}.
\end{equation}
Note that the bases $\{\psi_i\}_{i\geq 1}$
and $\{\phi_i\}_{i\geq 1}$ can be taken to be identical, but this is not necessary here.

We denote the orthogonal complement of $\Vnu$ by $\Vnu^\perp$. In terms of the basis 
functions, we have 
\begin{equation}
\Vnu^\perp=\overline{\spn\{\psi_{N+1},\psi_{N+2},\cdots\}}=\{\mbox{Mixed length scales}\ \ell<\ell_\nu\},
\end{equation} 
where the overline denotes 
closure in the $L^2$ topology. The space $\Vnu^\perp$ consists of functions that 
only contain the mixed scales, that is scales smaller than $\ell_\nu$ (see figure~\ref{fig:schem_scales}).

%With this basis, we can mathematically model the available and mixed length scales discussed in the previous section.
%We denote all scalar functions $f$ whose smallest length scale is larger than or equal to $\ell_I$ (respectively, $\ell_\nu$)
%by $\Vi$ (respectively, $\Vnu$). Since the basis $\{\phi_j\}_{j\geq 1}$ is ordered, there are positive integers $n$
%and $N$ ($n\ll N$) such that 
%\begin{align}
%\Vi= &\spn\{\phi_1,\phi_2,\cdots,\phi_n\}= \{\mbox{Initially available length scales}\ \ell\geq\ell_I \},\\
%\Vnu= &\spn\{\phi_1,\phi_2,\cdots,\phi_N\}=\{\mbox{Unmixed length scales}\ \ell\geq\ell_\nu \}.
%\end{align}
%
%We denote the orthogonal complement of $\Vnu$ by $\Vnu^\perp$. In terms of the basis 
%functions, we have 
%\begin{equation}
%V_{\ell_\nu}^\perp=\overline{\spn\{\phi_{N+1},\phi_{N+2},\cdots\}}=\{\mbox{Mixed length scales}\ \ell<\ell_\nu\},
%\end{equation} 
%where the overline denotes 
%closure in the $L^2$ topology. The space $\Vnu^\perp$ consists of functions that 
%only contain the mixed scales, that is scales smaller than $\ell_\nu$ (see figure~\ref{fig:schem_scales}).

\subsection{Main result}\label{sec:MainResult}
Given an initial condition $f\in \Vi$ for the tracer, its advected 
image $\mathcal Pf\in L^2$ at the final time
can potentially contain all length scales $\ell_j$.  
The flow redistributes the `energy' budget of the tracer among various scales in such a way
that the $L^2$-norm is conserved, i.e., 
\begin{equation}
\|f\|_{L^2}=\|\mathcal P f\|_{L^2}
=\underbrace{\sum_{i=1}^{N}|\langle \mathcal Pf,\psi_i\rangle_{L^2}|^2}_\text{unmixed} 
+ \underbrace{\sum_{i=N+1}^{\infty}|\langle \mathcal Pf,\psi_i\rangle_{L^2}|^2}_\text{mixed}.
\label{eq:parseval}
\end{equation}
A tracer is better mixed if more of its energy budget is transfered to the mixed scales $\ell <\ell_\nu$.
Therefore, we seek optimal initial conditions $f\in \Vi$ such that the energy budget of its image $\mathcal P f$ 
is mostly stored in the mixed scales, maximizing
$\sum_{i=N+1}^{\infty}|\langle \mathcal Pf,\psi_i\rangle|^2$. To make
these statements more precise we introduce the following truncation of the PF operator.

\begin{defn}[Truncated Perron--Frobenius operator]
We define the truncated PF operator $\mathcal P_p:\Vi\to \Vnu$ as 
the linear map $\mathcal P_p=\proj_N\circ\mathcal P$, where $\proj_N$ is the orthogonal projection
onto the $N$-dimensional subspace $\Vnu$. We also define the remainder operator 
$\mathcal P^\perp_p:\Vi\to \Vnu^\perp$
as $\mathcal P^\perp_p=\mathcal P-\mathcal P_p$.
\label{def:PFtrunc}
\end{defn}

It follows from Parseval's identity that 
$\|\mathcal Pf\|_{L^2}^2=\|\mathcal P_pf\|_{L^2}^2+\|\mathcal P_p^\perp f\|_{L^2}^2$
(see equation~\eqref{eq:parseval}).
The quantity $\|\mathcal P_pf\|_{L^2}^2$ represents the portion of the energy budget of the tracer 
that remains unmixed after advection to the final time $t_0+T$. 
The quantity $\|\mathcal P_p^\perp f\|_{L^2}^2$, on the other hand,
represents the portion of the tracer that is completely mixed.
We, therefore, seek initial conditions $f\in \Vi$ that maximize the
mixed energy budget $\|\mathcal P_p^\perp f\|_{L^2}^2$.

Since the truncated PF operator $\mathcal P_p$ is a linear transformation between finite-dimensional vector spaces 
$\Vi$ and $\Vnu$, it can be represented by a matrix $P_p\in\mathbb R^{N\times n}$.
More specifically, for any $f\in \Vi$, there are scalars $\{\alpha_1,\cdots,\alpha_n\}$ and $\{\beta_1,\cdots,\beta_N\}$ such that 
\begin{equation}
f=\sum_{j=1}^{n}\alpha_j\phi_j\quad \mbox{and}\quad \mathcal P_pf=\sum_{i=1}^{N}\beta_i\psi_i.
\label{eq:expansions}
\end{equation}
The matrix $P_p$ maps $\pmb\alpha=(\alpha_1,\cdots,\alpha_n)^\top$ into 
$\pmb{\beta}=(\beta_1,\cdots,\beta_N)^\top$, that is $\pmb\beta =P_p\,\pmb\alpha$.
It follows from elementary linear algebra that
the entries $[P_p]_{ij}$ of the matrix $P_p$ are given by
\begin{equation}
[P_p]_{ij}=\langle \mathcal P\phi_j,\psi_i\rangle_{L^2}, \quad i\in\{1,2,\cdots,N\},\quad j\in\{1,2,\cdots,n\}.
\label{eq:Kmatrix}
\end{equation} 

With this prelude, we can now state our main result.
\begin{thm}
Consider the function spaces $\Vi$ and $\Vnu$ and their associated truncated 
PF operator defined above.
The solution of 
$$\arg\max \|\mathcal P_p^\perp f\|_{L^2},$$ 
with the maximum taken over all $f\in \Vi$ with $\|f\|_{L^2}=1$, is given by
$\fopt=\sum_{j=1}^{n}\alpha_{j}\phi_j$,
where $\pmb{\alpha}=(\alpha_1,\alpha_2,\cdots,\alpha_n)^\top$ 
is a right singular vector of the truncated PF matrix~\eqref{eq:Kmatrix} corresponding to its 
smallest singular value.
\label{thm:main}
\end{thm}
\begin{proof}
Since $\| \mathcal P_p^\perp f\|_{L^2}^2 = \| \mathcal P f\|_{L^2}^2-\| \mathcal P_p f\|_{L^2}^2=1-\| \mathcal P_p f\|_{L^2}^2$,
maximizing $\| \mathcal P_p^\perp f\|_{L^2}^2$ is equivalent to minimizing $\| \mathcal P_p f\|^2$.
Since $f$ belongs to the subspace $\Vi$, 
the initial condition $f$ and its image $\mathcal P_pf$ can be expressed
by the series~\eqref{eq:expansions} with $\pmb{\beta}=P_p\,\pmb{\alpha}$.
Denoting the standard Euclidean norm by $|\cdot|$, 
we have $|\pmb{\alpha}|^2=\|f\|_{L^2}^2=1$ and $|\pmb{\beta}|^2=|P_p\,\pmb{\alpha}|^2=\| \mathcal P_p f\|_{L^2}^2$.
Therefore,
\begin{equation}
\min_{f\in \Vi, \|f\|=1}\| \mathcal P_p f\|_{L^2}=\min_{|\pmb{\alpha}|=1}|P_p\,\pmb{\alpha}|.
\end{equation}
The minimum on the right hand side is well-known to 
coincide with the smallest singular value of the matrix $P_p$~\citep{stewart98}.
The minimum is attained at the corresponding right singular vector of the matrix $P_p$.
This completes the proof.
\end{proof}

Once the PF matrix $P_p$ is formed, the evaluation of the optimal initial
condition $\fopt$, from the above theorem, is straightforward. We outline the 
computation of the truncated PF matrix $P_p$ in section~\ref{sec:numerics}.

\begin{rem}
Note that if the matrix $P_p$ is not full-rank, there are initial conditions
$f$ of the form~\eqref{eq:expansions} with $|\pmb\alpha|= 1$, 
such that $|P_p\,\pmb\alpha|=0$. Such initial conditions result in `perfect mixing'
since their advected image $\mathcal Pf$ belongs entirely to the mixed scales $\ell<\ell_\nu$, i.e.,
$\mathcal Pf\in V_\nu^\perp$.
In the examples studied in Section~\ref{sec:results}, such perfect finite-time mixing was not
observed. In other words, the matrices $P_p$ are full-rank in these examples.
\end{rem}

\begin{rem}
We emphasize that the truncated PF operator $\mathcal P_p$
is \emph{not} used as an approximation of the full PF operator $\mathcal P$. 
Instead, the truncation $\mathcal P_p$ followed naturally from the physical assumption that length scales
$\ell<\ell_\nu$ are completely mixed. 
As is clear from equation~\eqref{eq:Kmatrix}, to evaluate the
truncation $\mathcal P_p$, one still needs to utilize the full PF operator to evaluate the terms $\mathcal P\phi_j$.
%Obviously, to carry out numerical computations, the infinite-dimensional operator $\mathcal P$ needs to be 
%approximated with reasonable accuracy (see, e.g.,
%\cite{ulam64,dellnitz99,dellnitz01,Bollt02}).
\end{rem}

\section{Numerical implementation}\label{sec:numerics}
Numerical computation of the optimal initial condition $\fopt$ relies on the scale-dependent 
bases $\{\phi_i \}_{i\geq 1}$ and $\{\psi_i \}_{i\geq 1}$. 
For completeness, we discuss two such bases: the Fourier basis and
the Haar wavelet basis. Since the examples considered in Section~\ref{sec:results} below 
are defined on equilateral two-dimensional domains, $\mathcal D=[0,L]\times [0,L]$, we focus
on this special case. The generalization to the rectangular domain and 
to the three-dimensional case is straightforward.

\subsection{Fourier basis}
For periodic boundary conditions, it is natural to use the Fourier basis to 
define the spaces $\Vi$ and $\Vnu$. The orthonormal Fourier basis associated with the two-dimensional domain 
$\mathcal D=[0,L]\times [0,L]$ consist of functions $(1/L)\exp\left[\i(2\pi/L)(\vc k\cdot \vc x)\right]$ where 
$\vc k\in \mathbb Z^2$ denotes the wave vector.
The length scale associated to each Fourier mode is inversely proportional to 
the wave number, $|\vc k|\sim\ell^{-1}$.
We take the space of available initial scalar fields to be the functions
whose Fourier modes contain at most a prescribed wave number $k_I\sim\ell_I^{-1}$, 
i.e., 

\begin{equation}
\Vi=\spn\left\{\frac{1}{L}\exp\left[\i\frac{2\pi}{L}(\vc k\cdot \vc x)\right] : 
\vc k=(k_x,k_y)\in \mathbb Z^2, |k_x|\leq k_I, |k_y|\leq k_I\right\}.
\label{eq:VI_kI}
\end{equation}
Similarly, the space of unmixed scales $V_\nu$ is the functions
whose Fourier modes contain at most a prescribed wave number $k_\nu\sim\ell_\nu^{-1}\gg k_I$, i.e., 
\begin{equation}
\Vnu=\spn\left\{\frac{1}{L}\exp\left[\i\frac{2\pi}{L}(\vc k\cdot \vc x)\right] : 
\vc k=(k_x,k_y)\in \mathbb Z^2, |k_x|\leq k_\nu, |k_y|\leq k_\nu\right\}.
\label{eq:Vnu_knu}
\end{equation}
More generally, one could define the space $\Vi$ (and similarly $\Vnu$) with independent upper
bounds $k_{I_x}$ and $k_{I_y}$ on the wave-numbers $|k_x|$ and $|k_y|$, respectively. 
Since the domain is equilateral,
and for simplicity, we choose the same upper bounds in both directions, $k_I=k_{I_x}=k_{I_y}$.

Since the tracer concentration is real-valued, the complex conjugate basis functions in $\Vi$ and $\Vnu$ are redundant. 
Also, the basis with $\vc k=\vc 0$ (corresponding to constant functions) is unnecessary since we assumed
that the tracer has zero mean. Excluding these redundant functions, the effective dimension 
of the vector spaces $\Vi$ and $\Vnu$ are $n=\dim \Vi=2k_I(k_I+1)$ and $N=\dim \Vnu=2k_\nu(k_\nu+1)$, respectively.

\subsection{Wavelet basis}
While the above Fourier basis is a convenient choice, it restricts its applicability to the periodic boundary conditions.
More general boundary conditions can be handled with an alternative basis, such as 
Haar wavelets. Such wavelet bases have the
added advantage that they can be localized in space in addition to scale. This property renders wavelets particularly
attractive in applications where the tracer can only be released into a subset of the fluid domain $\mathcal D$
due to geometric or design constraints. Contrast this with the global nature of the Fourier basis.
\begin{figure}
\centering
\includegraphics[width=.85\textwidth]{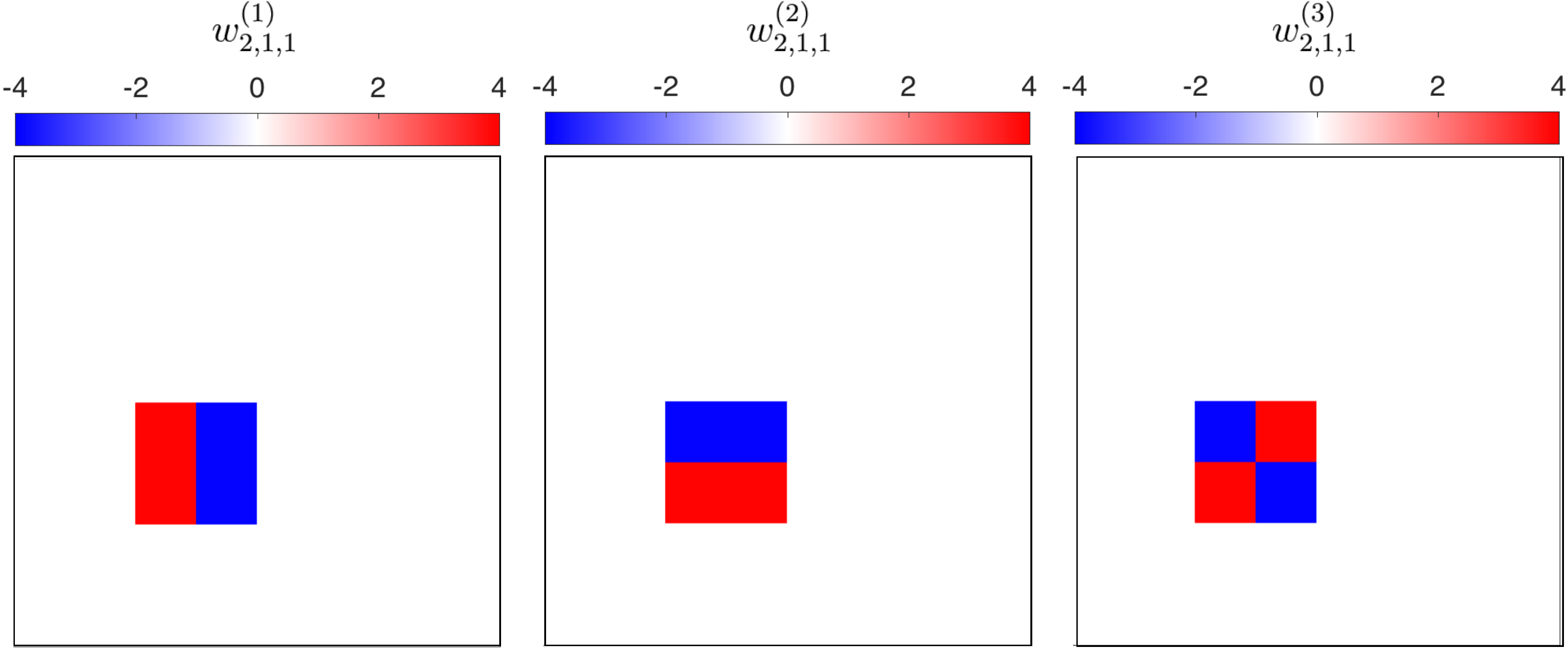}
\caption{Three examples of the wavelet functions~\eqref{eq:wavelet_2d_02} with
$j=2$, $i_x=i_y=1$. The domain is $\mathcal D=[0,1]\times[0,1]$}
\label{fig:wavelets}
\end{figure}

Here, we consider the Haar wavelet basis. For completeness, we briefly review the construction of this basis in 
two dimensions. Denote the one-dimensional Haar scaling function
with $s(x)=\mathds 1_{[0,1)}(x)$ 
and the corresponding wavelet with $h(x)=\mathds 1_{[0,1/2)}(x)-\mathds 1_{[1/2,1)}(x)$ where $\mathds 1_A$ is the 
indicator function of the set $A$. By dilations and translations, we obtain
\begin{subequations}
\begin{equation}
s_{j,i}(x)=2^{j/2}s\left(2^j\frac{x}{L}-i\right),\quad j\geq 0,\quad i\in\{0,1,\cdots, 2^j-1\},
\label{scaling_1d}
\end{equation}
\begin{equation}
h_{j,i}(x)=2^{j/2}h\left(2^j\frac{x}{L}-i\right),\quad j\geq 0,\quad i\in\{0,1,\cdots, 2^j-1\},
\label{wavelet_1d}
\end{equation}
\end{subequations}
where $0\leq x\leq L$ for a domain of size $L$.
The collection of the wavelets $h_{j,i}$ forms an orthogonal basis for mean-zero functions in 
$L^2([0,L])$~\citep{daubechies92}. 
The integer $j$ determines the size of the support of $h_{j,i}$ (or $s_{j,i}$) which is $L\times 2^{-j}$. Since
the wavelets with larger $j$ resolve finer structures (or smaller length scales), the 
integer $j$ is referred to as the scale of the wavelet. The integer $i$,
on the other hand, introduces a translation in the support of each wavelet, introducing a space dependence 
at each scale $j$.

The functions $s_{j,i}$ and $h_{j,i}$ serve as the building blocks of 
multidimensional wavelet bases~\citep{daubechies92,farge1999}. For instance, a
complete orthonormal basis for mean-zero functions in $L^2(\mathcal D)$, with $\mathcal D=[0,L]\times [0,L]$,
is formed by the set of functions 
\begin{align}
\big\{w^{(\mu)}_{j,i_x,i_y}:\ & 1\leq \mu \leq 3,\ 0\leq j,\  0\leq i_x\leq 2^{j}-1,\  0\leq i_y\leq 2^{j}-1 \big\},
\label{eq:wavelet_2d}
\end{align}
where 
\begin{subequations}
\begin{equation}
w_{j,i_x,i_y}^{(1)}(x,y)= \frac{1}{L}h_{j,i_x}(x)s_{j,i_y}(y),
\end{equation}
\begin{equation}
w_{j,i_x,i_y}^{(2)}(x,y)= \frac{1}{L}s_{j,i_x}(x)h_{j,i_y}(y),
\end{equation}
\begin{equation}
w^{(3)}_{j,i_x,i_y}(x,y)= \frac{1}{L}h_{j,i_x}(x)h_{j,i_y}(y).
\end{equation}
\label{eq:wavelet_2d_02}
\end{subequations}
The prefactor $1/L$ ensures that each basis function is of unit norm. The integer $j$ determine the 
scale in both $x$ and $y$ directions, while the integers $i_x$ and $i_y$ introduce the corresponding translations.
Figure~\ref{fig:wavelets} shows three examples of the two-dimensional wavelet functions~\eqref{eq:wavelet_2d_02}
with $j=2$.
The construction of two-dimensional wavelet bases from one-dimensional wavelets is not unique. 
For an alternative wavelet basis see, e.g., Chapter 10 of~\cite{daubechies92}.

Using the wavelet basis~\eqref{eq:wavelet_2d}, we define the subspace of initial conditions $\Vi$ as
\begin{align}
\Vi=\spn\big\{w^{(\mu)}_{j,i_x,i_y}:\ & 1\leq \mu\leq 3,\ 0\leq j\leq J_I-1 \nonumber\\
 & 0\leq i_x\leq 2^{j}-1,\  0\leq i_y\leq 2^{j}-1 \big\},
\label{eq:VI_wlet}
\end{align}
where the integer $J_I$ sets the initially available length scales.
Roughly speaking, the wavelet subspace $V_I$ 
contains tracer blobs of size $\ell_I=L\times 2^{-J_I}$ or larger.
Similarly, we define the subspace of unmixed length scales by
\begin{align}
\Vnu=\spn\big\{w^{(\mu)}_{j,i_x,i_y}:\ & 1\leq \mu\leq 3,\ 0\leq j\leq J_\nu-1 \nonumber\\
 & 0\leq i_x\leq 2^{j}-1,\  0\leq i_y\leq 2^{j}-1 \big\},
\label{eq:Vnu_wlet}
\end{align}
containing the unmixed tracer blobs of size $\ell_\nu=L\times 2^{-J_\nu}$ or larger.
For given positive integers $J_I$ and $J_\nu$, we have $n=\dim \Vi=4^{J_I}-1$ and $N=\dim \Vnu =4^{J_\nu}-1$.

Recall that the basis functions $\phi_i$ spanning the domain $\Vi$ of the truncated PF operator $\mathcal P_p$
need not to be identical to the basis functions $\psi_i$ spanning its range $\Vnu$. As a result,
the Fourier-based subspaces~\eqref{eq:VI_kI} and~\eqref{eq:Vnu_knu} can be used 
in conjunction with the wavelet-based subspaces~\eqref{eq:VI_wlet} and ~\eqref{eq:Vnu_wlet}.
In the following, we consider examples with both Fourier-based and 
wavelet-based subspaces~\eqref{eq:VI_kI} and~\eqref{eq:VI_wlet} for defining the domain $\Vi$.
For the range $\Vnu$, however, we only consider the Fourier-based subspace~\eqref{eq:Vnu_knu}
in order to achieve speedup in the computations by taking advantage of the fast Fourier transform \texttt{FFTW}.
\\

Once the choice of bases is made, the truncated PF matrix~\eqref{eq:Kmatrix} can be computed
by evaluating the integral,
\begin{equation}
[P_p]_{ij}=\langle \mathcal P\phi_j,\psi_i \rangle_{L^2}:=\int_{\mathcal D}(\mathcal P\phi_j)(\vc x)[\psi_i(\vc x)]^\ast\id \vc x,
\label{eq:Kmatrix_2}
\end{equation}
where $\ast$ denotes the complex conjugation. We approximate this integral using
the standard trapezoidal rule~\citep{recipe}. To ensure the accuracy of the approximation, 
the results reported in section~\ref{sec:results} are computed
using a dense uniform grid $\mathcal G$ of $2048\times 2048$ collocation points over 
the domain $\mathcal D=[0,L]\times [0,L]$. The terms $\mathcal P\phi_j$ are computed from 
the definition of the PF operator (Definition~\ref{def:PF}), i.e.,
$(\mathcal P\phi_j)(\vc x_0)=\phi_j(\fm^{-1}(\vc x_0))$ for any
$\vc x_0\in\mathcal G$. (see~\cite{dellnitz01}, for more accurate numerical methods).

\section{Examples and discussion}\label{sec:results}
\subsection{A time-periodic model}\label{sec:sineMap}
As the first example, we consider the time-periodic sine flow~\citep{liu94,pierrehumbert94}. 
This model is simple enough to unambiguously demonstrate our results, yet it can exhibit complex 
dynamics with simultaneous presence of chaotic mixing and coherent vortices.

The sine flow has a spatially sinusoidal velocity field on the domain 
$(x,y)\in[0,1]\times [0,1]$ with periodic boundary conditions. 
The temporal period of the flow is $2\tau$ for some $\tau>0$.
During the first $\tau$ time units, the velocity field is $\vc u=(0,\sin(2\pi x))^\top$
and switches instantly to $\vc u=(\sin(2\pi y),0)^\top$ for the second $\tau$
time units. This process repeats iteratively.
\begin{figure}
\centering
\includegraphics[width=.7\textwidth]{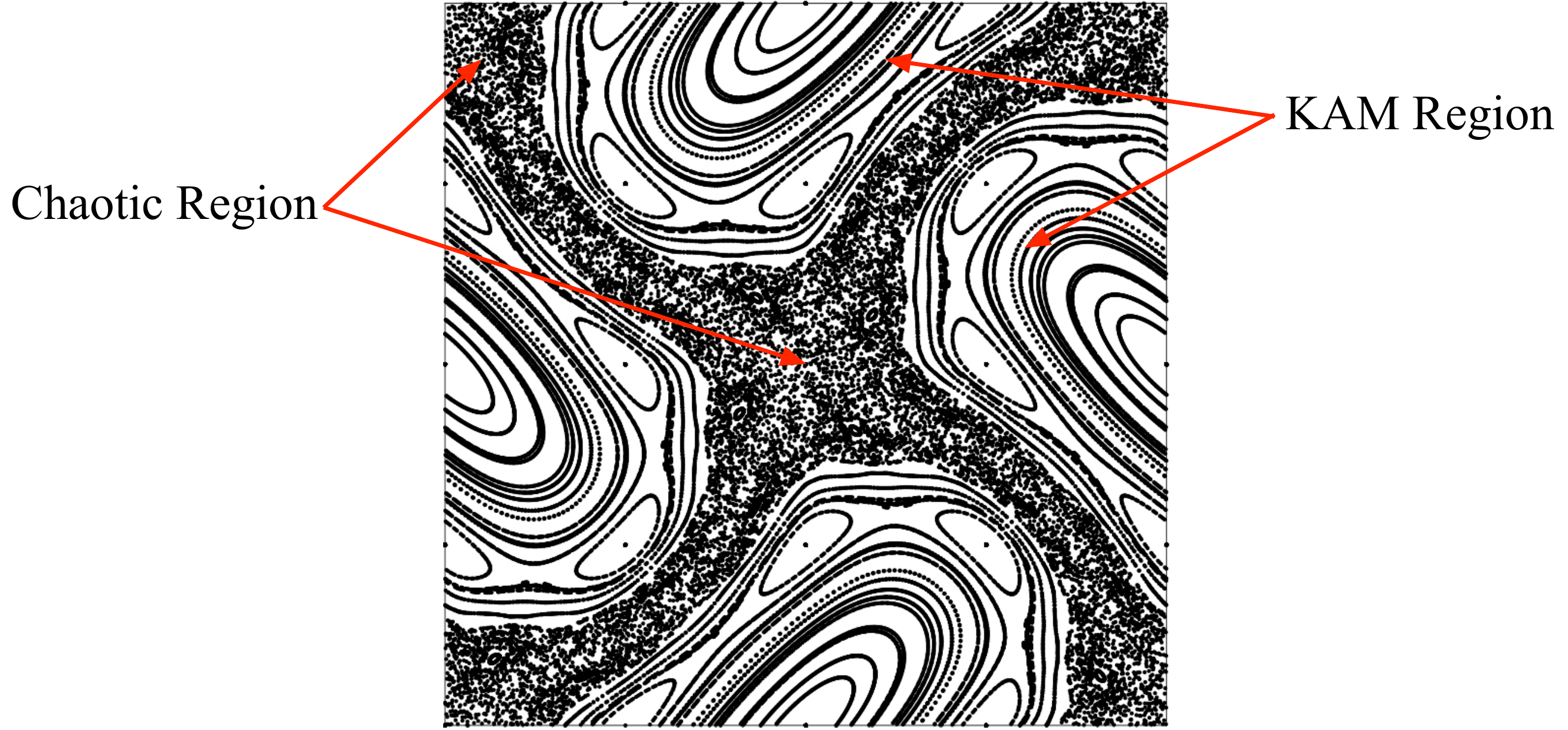}
\caption{Two hundred iterations of the sine map~\eqref{eq:map} with $\tau=0.25$ starting from $16\times 16$
initial conditions distributed uniformly over the domain $[0,1]\times [0,1]$.
The domain is periodic in both directions.}
\label{fig:pmap}
\end{figure}

The sine flow generates a reversible map $T$ that, over one period, maps 
points $(x,y)$ to $T(x,y)$. The inverse of the map $T$ is given explicitly by~\citep{gubanov10}
\begin{equation}
T^{-1}:
\begin{pmatrix}
x\\ y
\end{pmatrix}
\mapsto
\begin{pmatrix}
x-\tau \sin(2\pi y)\\ 
y-\tau \sin\big(2\pi (x-\tau \sin(2\pi y))\big)
\end{pmatrix}
\mod{1}.
\label{eq:map}
\end{equation}

Figure~\ref{fig:pmap} shows $200$ iterations of this map
with $\tau=0.25$ launched from a uniform grid of initial conditions. The map
$T^{-1}$ has two hyperbolic fixed points located at $(0,0)$ and $(0.5,0.5)$ whose
tangle of stable and unstable manifolds creates a chaotic mixing region.
In addition, the map has two elliptic fixed points located at $(0.5,0)$ and
$(0,0.5)$. These elliptic fixed points are surrounded by invariant 
Kolmogorov--Arnold-Moser (KAM) tori with quasi-periodic motion that inhibit mixing~\citep{topolHydro_arnold}.
\begin{figure}
\centering
\includegraphics[width=\textwidth]{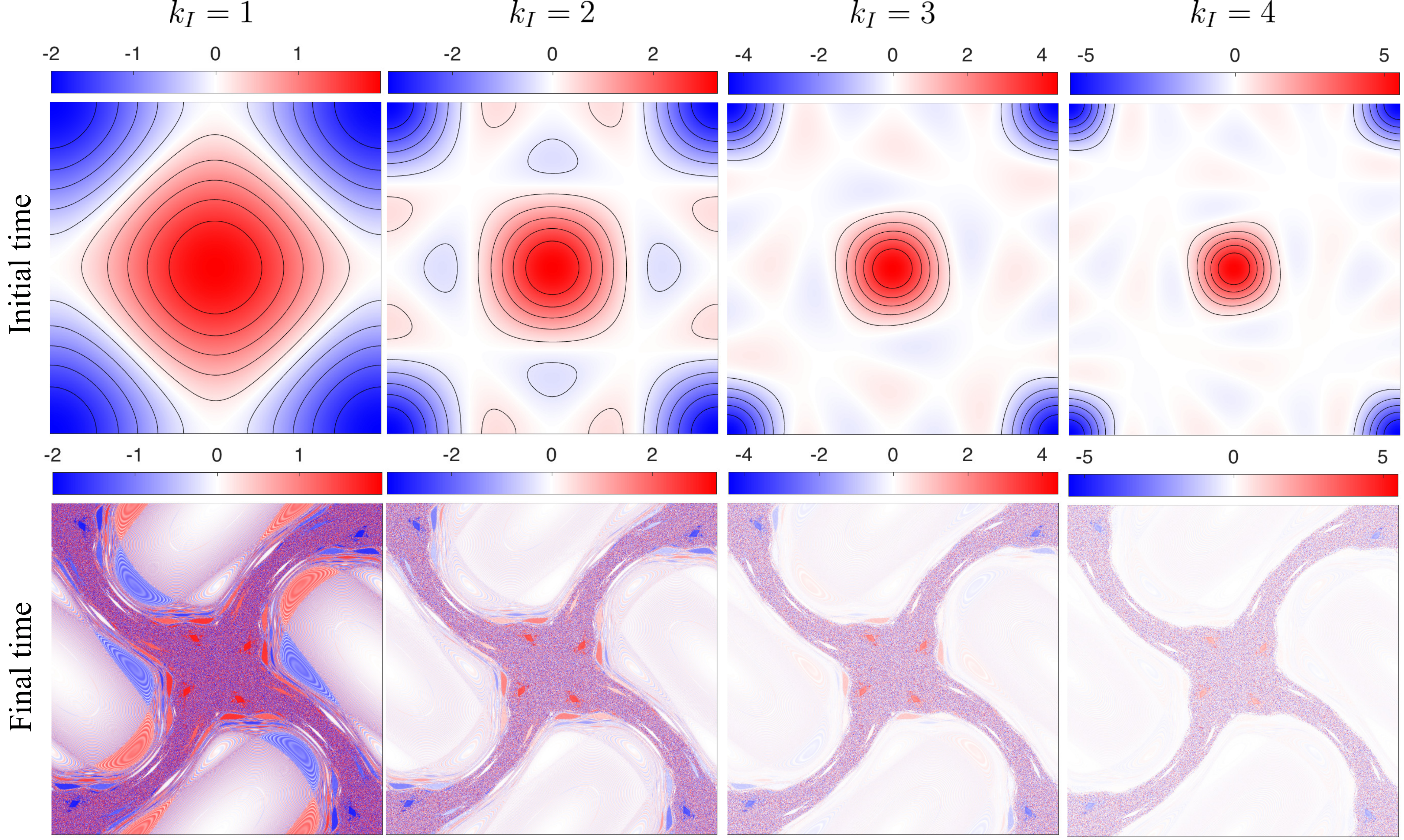}
\caption{The optimal initial conditions $\fopt\in\Vi$ for the sine map, with
$\Vi$ being the Fourier-based subspace defined in~\eqref{eq:VI_kI}.
Four optimal initial conditions with $k_I=1$, $2$ , $3$ and $4$
are shown in the top panel. The range of the truncated PF operator 
$\mathcal P_p$ is the Fourier-based subspace $\Vnu$ with $k_\nu=256$.
The bottom panel shows their corresponding
advected image $\mathcal P\fopt$ under $200$ iterations of the sine map.
All figures show the entire domain $\mathcal D=[0,1]\times [0,1]$.}
\label{fig:optIC}
\end{figure}

It is known that mixing is more efficient around the hyperbolic fixed points due to their
tangle of stable and unstable manifolds~\citep{aref}. The KAM regions, in contrast, form islands of coherent motion 
that inhibit efficient mixing of passive tracers. Therefore, it is desirable to release the tracer blobs around the hyperbolic fixed points, avoiding
the KAM region. Here, we examine whether the optimal initial condition given by Theorem~\ref{thm:main}
agrees with this intuitive assessment.
\begin{figure}
\centering
\includegraphics[width=.75\textwidth]{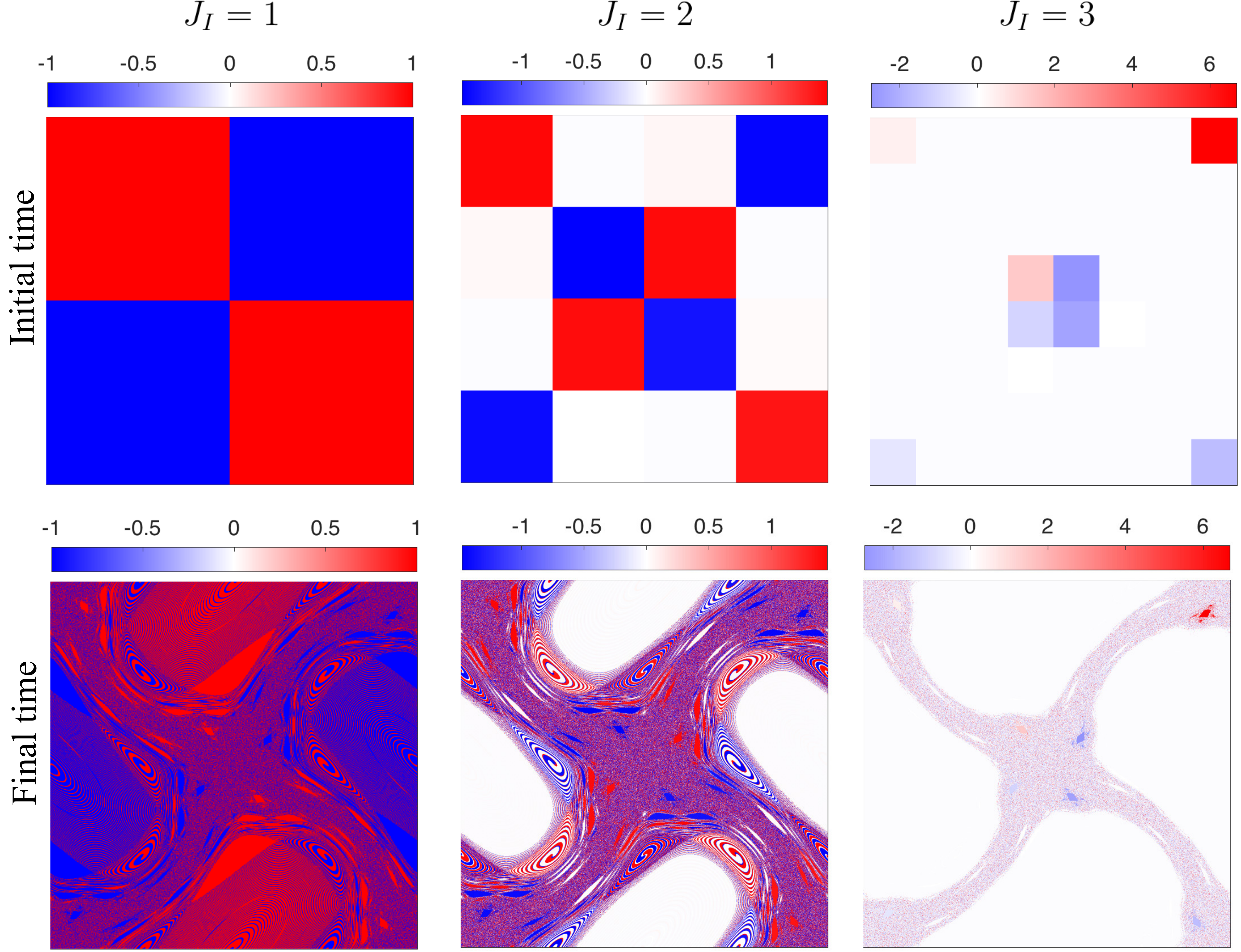}
\caption{The optimal initial conditions $\fopt\in\Vi$ for the sine map.
Four optimal initial conditions with $J_I=1$, $2$ , $3$
are shown in the top panel. The bottom panel shows their corresponding
advected images $\mathcal P\fopt$ under $200$ iterations of the sine map.
The range of the truncated PF operator 
$\mathcal P_p$ is the Fourier-based subspace $\Vnu$ with $k_\nu=256$.
All figures show the entire domain $\mathcal D=[0,1]\times [0,1]$.}
\label{fig:optIC_wavelet}
\end{figure}

For the finite-time analysis, we consider the flow under $200$ iterations
of the sine map, i.e. we set the flow map $\fm = T^{200}$. 
First, we consider the Fourier-based initial subspace $\Vi$ defined in~\eqref{eq:VI_kI}.
Figure~\ref{fig:optIC} shows the optimal initial conditions obtained from Theorem~\ref{thm:main}
with $k_I=1,2,3$ and $4$. For all parameter values $k_I$, the optimal initial condition consists of two
prominent blobs centered at the hyperbolic fixed points $(0,0)$ and $(0.5,0.5)$.
For $k_I=1$, only very large scales are available for the distribution of the tracer blob
and therefore some intersection with the KAM region is inevitable.
As the number of available wave numbers $k_I$ (or equivalently, available initial length scales)
increases the blobs become more concentrated at the hyperbolic fixed points. 

Even for $k_I=4$, the optimal initial condition has very small but non-zero concentration 
in the KAM regions. This is due to the global nature of the Fourier modes which 
inhibits the perfect localization around the hyperbolic fixed points. 
The wavelet-bases subspace~\eqref{eq:VI_wlet} does not suffer from this drawback. 
Figure~\ref{fig:optIC_wavelet}, for instance, shows three optimal initial conditions
in this wavelet-based subspace. 
For $J_I=1$, where only the largest scales are available, intersection
with the KAM region is inevitable (similar to the case of $k_I=1$ in figure~\ref{fig:optIC}).
As the smaller scales become available, the optimal initial condition $\fopt$ concentrates
around the hyperbolic fixed points with no concentration at the KAM regions.

The results in figures~\ref{fig:optIC} and \ref{fig:optIC_wavelet} are computed using the
Fourier-based subspace $\Vnu$ with $k_\nu=256$. To ensure the insensitivity
of the results to perturbations, we recomputed them
by varying the cut-off wavenumber in the interval $250\leq k_\nu\leq 260$ and obtained 
almost identical optimal initial conditions.

\subsection{Two-dimensional turbulence}\label{sec:2Dturb}
As the second example, we consider a fully unsteady flow obtained from a direct numerical 
simulation of the two-dimensional Navier--Stokes equation,
\begin{equation}
\partial_t \vc u +\vc u\cdot \bnabla\vc u = -\bnabla p +\nu \Delta \vc u +\vc F,\quad \bnabla\cdot\vc u=0,
\end{equation}
with the dimensionless viscosity $\nu=10^{-5}$ and a band-limited stochastic forcing $\vc F$. 
The flow domain is the box $\mathcal D=[0,2\pi]\times[0,2\pi]$
with periodic boundary conditions. A standard pseudo-spectral code 
with 2/3 dealiasing was used to numerically solve the Navier--Stokes equations
(see Section 6.2 of~\citet{pra} for further computational details).
\begin{figure}
\centering
\includegraphics[width=.95\textwidth]{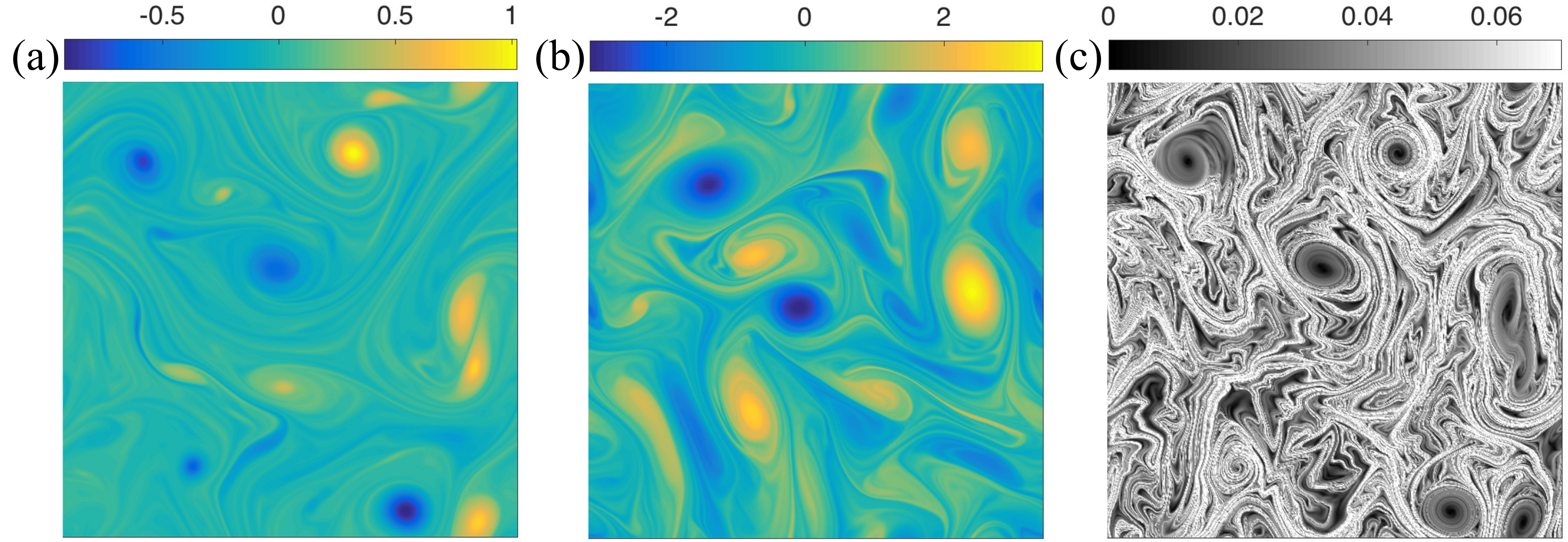}
\caption{(a) The vorticity field at the initial time $t_0$. (b) The vorticity field at the final time $t_0+T$.
(c) The forward-time FTLE field corresponding to the time interval $[t_0,t_0+T]$. Here, $T$
is $100$ time units.}
\label{fig:2Dturb_vort}
\end{figure}

Starting from a random-phase initial
condition, we numerically integrate the Navier--Stokes equation. After
$1000$ time units the flow has reached a statistically steady turbulent state with Reynolds number $4.1\times 10^3$. 
We set this time as the initial time $t_0$ for the
mixing analysis. The final time instance is set to $t_0+T$ with $T=100$.
Figures~\ref{fig:2Dturb_vort}(a,b) show the vorticity fields at these initial and final times. 

As is typical of two-dimensional turbulence, the flow contains several coherent vortices that 
exhibit minimal material deformation over the time interval $[t_0,t_0+T]$~\citep{mcwilliams1984vortices}. These coherent
vortices are signaled by the islands of small finite-time Lyapunov exponent (FTLE)
shown in figure~\ref{fig:2Dturb_vort}(c). The FTLE field is computed as
$\log [\lambda(\vc x)]/(2T)$ for all $\vc x\in\mathcal D$, with $\lambda(\vc x)$ being the largest 
eigenvalue of the Cauchy--Green strain tensor $[\id \fm(\vc x)]^\top\id\fm (\vc x)$, and $\id\fm$ 
denoting the Jacobian of the flow map~\citep{voth02}.
Outside the coherent vortices the flow is mostly chaotic, dominated by the stretching
and folding of material lines.
\begin{figure}
\centering
%\includegraphics[width=.24\textwidth]{turb_T100_n2048_ki1kf512_f}
%\includegraphics[width=.24\textwidth]{turb_T100_n2048_ki2kf512_f}
%\includegraphics[width=.24\textwidth]{turb_T100_n2048_ki3kf512_f}
%\includegraphics[width=.24\textwidth]{turb_T100_n2048_ki4kf512_f}\\
%\includegraphics[width=.24\textwidth]{turb_T100_n2048_ki1kf512_Kf}
%\includegraphics[width=.24\textwidth]{turb_T100_n2048_ki2kf512_Kf}
%\includegraphics[width=.24\textwidth]{turb_T100_n2048_ki3kf512_Kf}
%\includegraphics[width=.24\textwidth]{turb_T100_n2048_ki4kf512_Kf}
%\hspace{.18\textwidth}
\includegraphics[width=\textwidth]{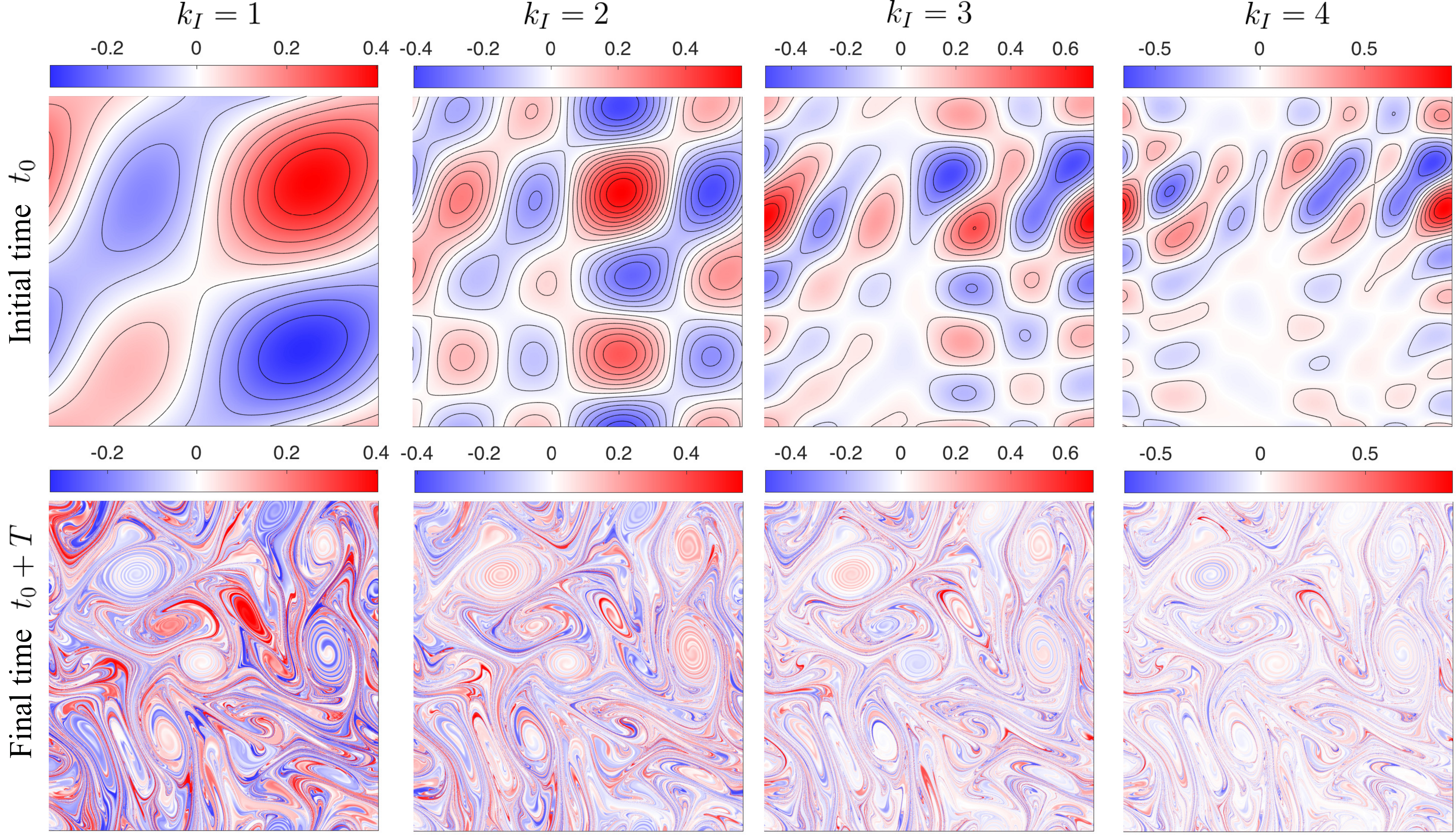}
\caption{Optimal initial tracers $\fopt\in \Vi$ (upper panel) in the turbulent flow with $k_I=1,2,3$ 
and $4$. Their corresponding advected images $\mathcal P\fopt$ at the final time $t_0+T$ are 
shown in the lower panel. In all four cases shown here, the mixed wavenumber is $k_\nu=256$ (see equation~\eqref{eq:Vnu_knu}).
All panels show the entire domain $\mathcal D=[0,2\pi]\times [0,2\pi]$.}
\label{fig:2Dturb}
\end{figure}

Next we compute the optimal initial conditions $\fopt$.
Unlike the sine map, the preimages $\fm^{-1}(\vc x_0)$ are not explicitly known
here. We numerically evaluate the preimages by integrating
the ODE~\eqref{eq:ode} backwards in time from the final time $t_0+T$ to the initial time $t_0$,
for each initial condition $\vc x_0\in\mathcal G$.
This numerical integration is carried out by the fifth-order Runge-Kutta scheme
of~\cite{RK45}. Since the velocity field $\vc u$ is stored on a discrete spatiotemporal 
grid, it needs to be interpolated for the particle advection. 
Here, we use cubic splines for the spatial interpolation of the velocity field 
together with a linear interpolation in time.
\begin{figure}
\centering
\includegraphics[width=.6\textwidth]{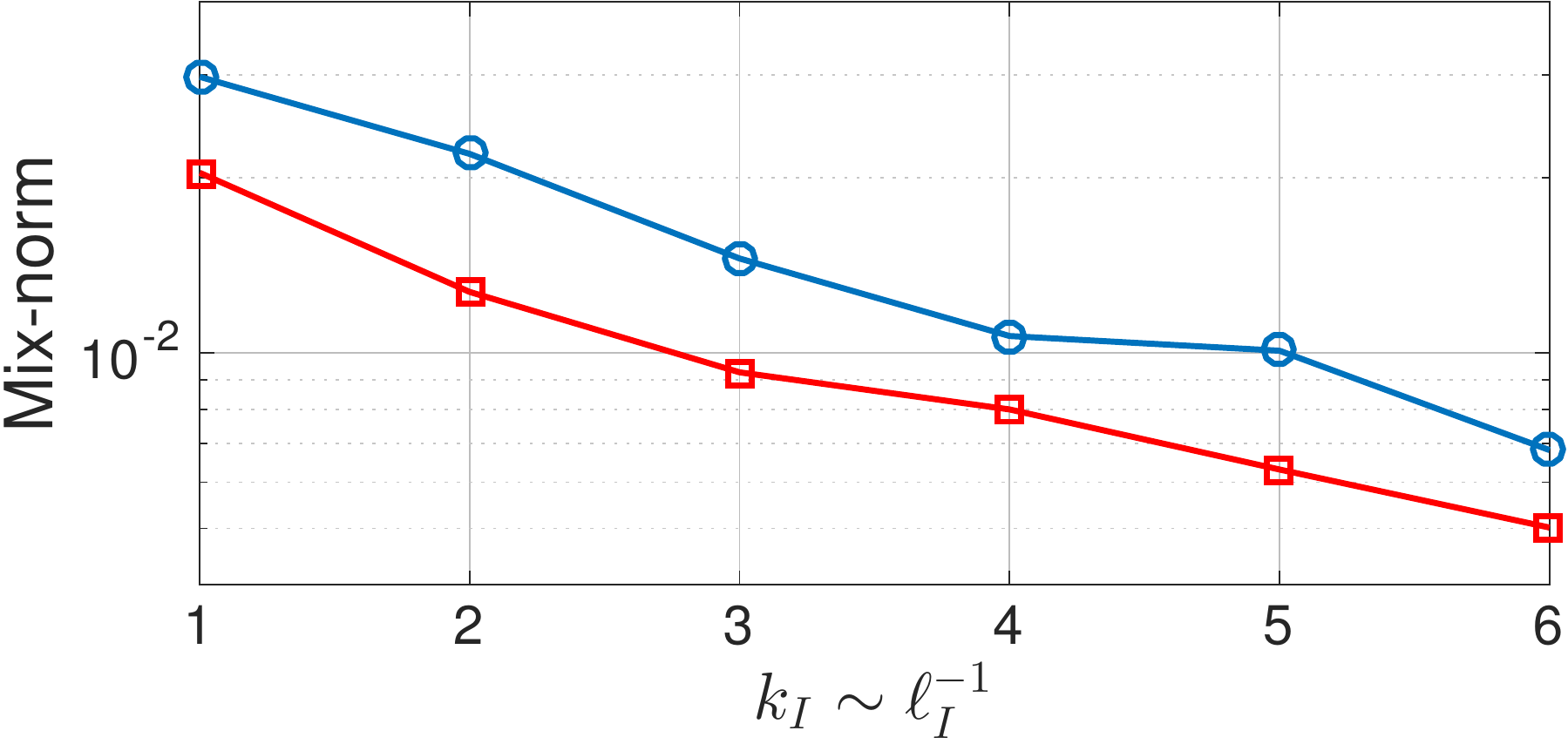}
\caption{Red line (squares): The mix-norm $\|\mathcal P\fopt\|_{H^{-1}}$ 
for the optimal initial concentrations $\fopt\in \Vi$. The optimal initial conditions for $k_I=1$, $2$, $3$ and $4$
are shown in figure~\ref{fig:2Dturb}.
Blue line (circles): The mix-norm $\|\mathcal Pf\|_{H^{-1}}$ for the initial concentrations 
$f(x,y)=\cos(k_Ix)\cos(k_Iy)/\pi$.
\label{fig:mixNorm}
}
\end{figure}

Figure~\ref{fig:2Dturb} shows the optimal initial tracer patterns $\fopt$ for 
$k_I=1,2,3$ and $4$, which belong to the corresponding Fourier-based subspaces $\Vi$
as defined in equation~\eqref{eq:VI_kI}. 
As opposed to the simple model considered in Section~\ref{sec:sineMap},
the optimal tracer patterns here have fairly complicated structures. This is to be expected 
as the turbulent flow itself has a complex spatiotemporal structure.

Ideally, the tracer should concentrate outside the coherent vortices to achieve 
better mixing. Similar to the sine flow, for $k_I=1$, where only the very large scales are available 
for the release of the tracer, there is some inevitable overlap between the coherent vortices and 
the tracer. This results in the visibly unmixed blobs in the advected tracers $\mathcal P\fopt$
shown in the lower panel of figure~\ref{fig:2Dturb}. Theorem~\ref{thm:main}, however, guarantees
that the optimal initial condition $\fopt$ is such that
the unmixed blobs are minimal. 
As smaller scales become available ($k_I>1$), the intersection of the high initial tracer concentration
and the coherent vortices becomes smaller, leading to a more homogeneous mixture after
advection to the final time $t_0+T$.
%While the FTLE field
%does distinguish between chaotic and coherent
%regions, it does not determine the optimal configuration of the initial tracer patterns.
%
\begin{figure}
\centering
\includegraphics[width=.75\textwidth]{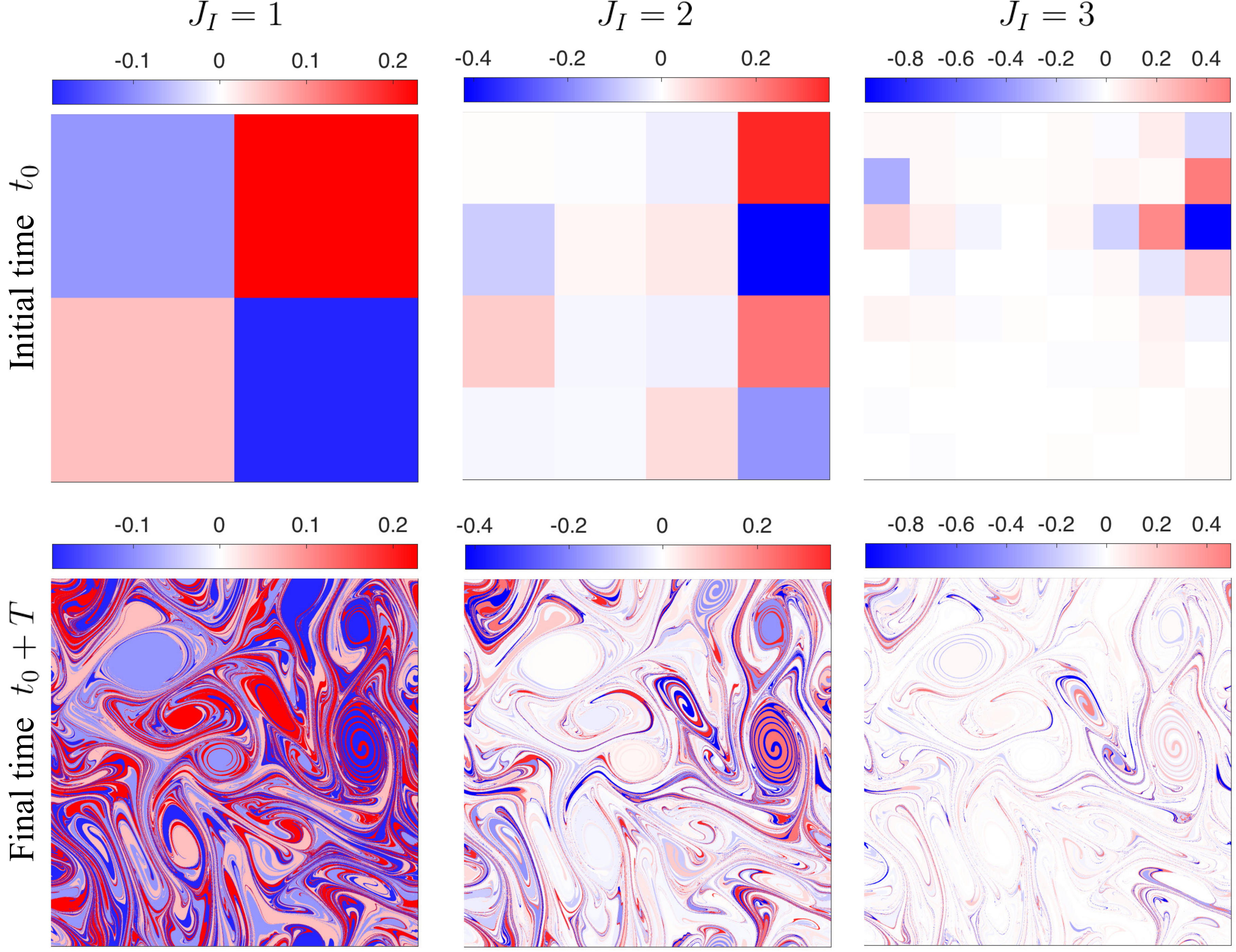}
\caption{Optimal initial tracers $\fopt\in \Vi$ (upper panel) in the turbulent flow with $J_I=1,2$ and $ 3$. 
Their corresponding advected images $\mathcal P\fopt$ at the final time $t_0+T$ are 
shown in the lower panel. In all four cases shown here, the mixed wavenumber is $k_\nu=256$ (see equation~\eqref{eq:Vnu_knu}).
All panels show the entire domain $\mathcal D=[0,2\pi]\times [0,2\pi]$.}
\label{fig:2Dturb_wavelet}
\end{figure}

To quantify the mixture qualities, we compute the mix-norm
of the advected tracers proposed by~\cite{shaw07}. This mix-norm is the Sobolev $H^{-1}$ norm, 
\begin{equation}
\|\rho\|_{H^{-1}}=\sqrt{\sum_{\vc k\neq \vc 0}|\widehat{\rho}(\vc k)|^2/|\vc k|^2},
\label{eq:mix-norm}
\end{equation}
where the hat sign denotes the Fourier transform. 
\cite{mathew07} proposed the alternative Sobolev norm $H^{-1/2}$ for quantifying the mixture quality.
The motivation for using such Sobolev norms is that the density of homogeneous mixtures
are concentrated at ever smaller scales or equivalently larger wave numbers $|\vc k|$.
As a result, more homogeneous mixtures have smaller mix-norms. 

The mix-norm $\|\mathcal P\fopt\|_{H^{-1}}$ is shown
in figure~\ref{fig:mixNorm} for the optimal initial conditions $\fopt\in\Vi$ with $1\leq k_I\leq 6$.
As $k_I$ increases, more homogeneous mixtures are obtained, as is also visible in figure~\ref{fig:2Dturb}.
For comparison, we also show the mix-norm $\|\mathcal Pf\|_{H^{-1}}$ for 
the non-optimal initial conditions $f(x,y)=\cos(k_Ix)\cos(k_Iy)/\pi\in\Vi$. The
non-optimal initial conditions result into a larger mix-norm, showing that they 
do not mix as well as the optimal initial conditions $\fopt$ do. 
Figure~\ref{fig:optIC_wavelet} shows the optimal initial conditions found
in the wavelet-based subspace~\eqref{eq:VI_wlet}. Their mix-norms exhibit a 
similar behavior as the one shown in figure~\ref{fig:mixNorm}.

\section{Concluding remarks}
The design of mixing devices has primarily been concerned with the 
stirring protocols that enhance mixing. The optimality of these protocols are limited by
design constraints and fundamental physics~\citep{lin11}. On the other hand, for a given stirring protocol, 
the final mixture quality also depends on the initial configuration of the 
tracer. The optimal initial condition for the release of the tracer has received far less attention.

Here, we proposed a rigorous framework for determining the
optimal initial tracer configuration to achieve maximal mixing under finite-time passive advection. 
We showed that, under reasonable assumptions, the problem reduces to a finite-dimensional 
optimization problem. The optimal initial condition then coincides with 
a singular vector of a truncated Perron--Frobenius (PF) operator.
This truncation is not an approximation of the infinite-dimensional PF operator; it rather 
follows naturally from our simplifying assumption that the tracer blobs smaller than
a prescribed critical length scale $\ell_\nu$ are completely mixed.

We discussed two numerical implementations of the 
optimization problem using Fourier modes and Haar wavelets.
While the Fourier modes are convenient for the spatially periodic flows considered
here, the wavelets are more suitable for handing more complicated geometries 
and boundary conditions. Wavelets also allow for optimal initial conditions
that are local in both space and scale. The space localization is crucial in many
applications where the tracer can only be released into a subset of the flow domain
duo to geometric constraints.

We restricted our attention here to ideal passive tracers.
Future work will expand the framework 
to account for diffusion and the presence of sinks and sources.
Diffusion, in particular, dictates a dissipative length scale $\ell_\nu$ for mixed blobs which, 
in the absence of diffusion, was prescribed here in an ad-hoc manner.

\ \\
\textbf{Acknowledgments:} 
I would like to thank Daniel Karrasch for pointing out the correct terminology in Definition~\ref{def:PF}
and bringing a number of relevant references to my attention.
I am also grateful to Charles Doering, Gary Froyland, George Haller
and Christopher Miles for their comments on this manuscript.

\end{document}